\documentclass[final]{dmtcs-episciences}
\usepackage[utf8]{inputenc}

\usepackage{amsmath}
\usepackage{amsthm}
\usepackage{amsfonts}
\usepackage{stackrel}
\usepackage{diagbox}
\usepackage{vaucanson-g}
\newtheorem{theorem}{Theorem}
\newtheorem{example}[theorem]{Example}
\newtheorem{lemma}[theorem]{Lemma}
\newtheorem{remark}[theorem]{Remark}
\newtheorem{definition}[theorem]{Definition}
\newtheorem{problem}[theorem]{Problem}
\newtheorem{proposition}[theorem]{Proposition}
\newtheorem{corollary}[theorem]{Corollary}
\newtheorem{conjecture}[theorem]{Conjecture}
\newtheorem{question}[theorem]{Question}

\newcommand{\movearrow}{\longrightarrow}%chemarrow
\newcommand{\move}[3][]{#2\hspace*{0.1em plus 0.05em}{\movearrow_{#1}}\hspace*{0.1em plus 0.05em}#3}
\newcommand{\emove}[1]{\move{#1}{\varepsilon}}
\newcommand{\ta}{{\tt a}}
\newcommand{\tb}{{\tt b}}

\renewcommand{\P}{\mathcal P}
\newcommand{\N}{\mathcal N}
\newcommand{\g}{\mathcal G} % Grundy value
\DeclareMathOperator{\mex}{mex}
\title{Taking-and-merging games as rewrite games\thanks{Supported by the ANR-14-CE25-0006 project of the French National Research Agency and the CNRS PICS-07315 project.}}

\author{ Eric Duch\^ene\affiliationmark{1}
    \and Victor Marsault\affiliationmark{2}
    \and Aline Parreau\affiliationmark{1}
    \and Michel Rigo\affiliationmark{3}}
\affiliation{ LIRIS, Université Claude Bernard Lyon 1, CNRS, France \\
              LIGM, Univiversité Gustave Eiffel, CNRS, France \\
              Department of Mathematics, University of Liège, Belgique}

\received{2019-02-21}

\revised{2020-05-26}

\accepted{2020-08-03}

\begin{document}
\publicationdetails{22}{2020}{4}{5}{5200}

\maketitle

\begin{abstract}
    This work is a contribution to the study of rewrite games.
    Positions are finite words, and the possible moves are defined by a finite number of
    local rewriting rules $\{\move{u_i}{v_i}\}_{i\in I}$:
    a move consists in the substitution  of one occurrence of $u_i$ by~$v_i$, for some~$i$.
    We introduce and investigate taking-and-merging games, that is, where each rule is of the form
    $\emove{a^k}$.
    We give sufficient conditions for a game to be such that the losing positions (resp.\@ the positions with a given Grundy value) form a regular language or a context-free language. We formulate several related open questions in parallel with the famous conjecture of Guy about the periodicity of the Grundy function of octal games.

    Finally we show that more general rewrite games quickly lead to undecidable problems.
    Namely, it is undecidable whether there exists a winning position in a given regular language, even if we restrict to games where each move strictly reduces the length of the current position.
\end{abstract}

\section{Introduction}

Waldmann \cite{Waldmann:2002} introduces general {\em rewrite games} as follows. 
Let $A$ be a finite alphabet, i.e., a finite set of symbols. We let $A^*$ denote the set of finite words over~$A$. The empty word is denoted by $\varepsilon$. A rewrite system is given by a (finite) set $R\subset A^*\times A^*$ of rules, called $R$-reductions, of the form $\move{u}{v}$. The latter rule can be applied to the word $w=xuy$, $x,y\in A^*$  where we replace one occurrence of $u$ by $v$ and we write $\move[R]{w}{xvy}$. 
We consider only \emph{terminating} rewrite systems, that is, such that there is no infinite chain of~$R$-reductions starting from a given word.
In the rewrite game associated with~$R$, the positions are the words in~$A^*$, and from a position~$w$
the possible moves are those that lead to each word~$w'$ such that~$\move[R]{w}{w'}$.
Starting from a word, also called ground term, $t_1\in A^*$, two players apply alternatively an $R$-reduction of their choice to get a sequence $t_1\movearrow_R t_2 \movearrow_R t_3 \movearrow_R \cdots \movearrow_R t_n$ until no $R$-reduction can be applied. The first player unable to apply an $R$-reduction, because $t_n$ is in normal form (i.e., irreducible), loses the game ($t_n$ is called a \emph{final position} of the game). 

Rewrite games belong to the family of impartial combinatorial games.
In an impartial combinatorial game, two players move alternatively
with perfect information, and the set of valid moves depends only on the position.
The first player unable to move loses the game. 
A complete definition of combinatorial games can be found in~\cite{Siegel}. 
Taking-and-breaking games are famous examples of combinatorial games.
A position consists in several piles of tokens, and
a move consists in removing some tokens from a pile,
and then splitting that pile into smaller piles.
A major issue when studying combinatorial games is the computation of the \emph{outcome}. An impartial combinatorial game position has outcome $\N$ if the player who starts has a winning strategy, and $\P$ otherwise. 

The notion of \emph{Grundy value} (also called \emph{Sprague--Grundy value}
is a refinement of the one of \emph{outcome}: the position with outcome $\P$ are exactly
those whose Grundy value is~$0$.
More precisely, the Grundy value of any position is recursively defined as the $\mex$ (minimum excluded value) of the set of Grundy values of the position reachable in one move.
For example, $\mex \{0,1,3\}=2$, and by convention~$\mex \emptyset = 0$.

\medskip

A background motivation for this work stems from octal games. They are a well-known family of combinatorial games that can be described as rewrite games. 
They are the taking-and-breaking games in which it is never
allowed to split into more than two piles.
An octal game is defined by its valid moves, which may be coded by a (finite or infinite) sequence of integers that are less than or equal to $7$; see~\cite{Siegel} for a formal definition.
%Subtraction games (i.e., taking-and-breaking games where no breaking is allowed) may also be described as rewrite games. % VM: I don't think this is useful
Octal games can be translated as rewrite games as follows.
If we have $r$ piles of token with respectively $n_1,\ldots,n_r$ tokens, then a position in such a game can be coded by the word over a two-letter alphabet 
\begin{equation*}
    \tb\ta^{n_1}\tb\ta^{n_2}\tb \cdots \tb\ta^{n_r}\tb.
\end{equation*}
The $\tb$'s play the role of separators between piles of $\ta$'s and one has to carefully choose the convenient reductions to code the game of interest, see \cite[Prop.~3]{Waldmann:2002}. 

\begin{example}\label{exa:0.37}
Let us consider the game over the alphabet $A=\{\ta,\tb\}$, associated with the rewrite system $R=\{\emove{\ta},\emove{\ta\ta},\move{\ta\ta}{\tb}\}$. An example of sequence of play for this game, starting from the position $t_1=\tb\ta\ta\ta\tb\ta\ta\tb$, is
$$
\tb\ta\ta\ta\tb\ta\ta\tb \movearrow_R \tb\ta\ta\ta\tb\ta\tb  \movearrow_R \tb\ta\tb\ta\tb \movearrow_R \tb\tb\ta\tb \movearrow_R \tb\tb\tb~.
$$
In this example, four moves have been played, hence the second player wins the game. Note that this game exactly corresponds to the octal game $0.37$.
Indeed the piles are the block of one or more consecutive $\ta$'s.
A player can remove one token of a pile, possibly emptying it, by using move $\emove{\ta}$.
A player can also remove two tokens from a pile, possibly emptying it, by applying $\emove{\ta\ta}$.
Finally, a player can remove two tokens from a pile and divide the remaining tokens
into two piles, by applying $\move{\ta\ta}{\tb}$ in the middle of a block of $\ta$'s.
\end{example}

The \emph{Grundy sequence} of an octal game is defined as the integer sequence where the $i$-th element is the Grundy
value of the position with one pile of~$i$ tokens.
We may then reformulate a famous conjecture in combinatorial game theory:

\begin{conjecture}[Guy's conjecture \cite{BCG}]\label{conj:guy}
 All finite octal games have an eventually periodic Grundy sequence.
\end{conjecture}

In the context of a rewrite game~$G$, positions are words over a finite alphabet $A$ and we can associate a Grundy value $\g(w)$ with each word $w$ in $A^*$.
Thus, the family of languages $(\mathcal{L}_i)_{i\in\naturals}$, defined by~$\mathcal{L}_i = \g^{-1}(i)$,
is a partition of~$A^*$; they are called the \emph{Grundy languages} of~$G$.
In his paper \cite{Waldmann:2002}, Waldmann makes a correspondence between the regularity
of the Grundy languages of octal games (seen as rewrite games) and the periodicity of the Grundy sequence.

\begin{theorem}[Waldmann, 2002]\label{thm:waldmann}
The Grundy sequence of an octal game is eventually periodic if and only if it has only finitely many non-empty Grundy languages $\mathcal{L}_i$, all of which are regular languages. 
\end{theorem}

This nice result translates the notion of periodicity of a taking-and-breaking game into the context of rewrite games. Therefore, the question of the regularity of rewrite games becomes paramount,
and in particular would allow to make progress towards proving or disproving Conjecture~\ref{conj:guy}.
This leads to the general open question below, which we start to address in this article.

\begin{question}
Which rewrite games have Grundy values bounded by a constant $K$ and such that all the languages $\mathcal{L}_0,\ldots,\mathcal{L}_K$ are regular?
\end{question}

The following classical lemma (see \cite{BCG}) characterizes the Grundy languages of a rewrite game;
and will be heavily used throughout this article.

\begin{lemma}\label{lem:Grundy}
Given a rewrite game $G$ over an alphabet $A$, 
the family~$(\mathcal{L}_i)_{i\in\naturals}$ is the only family of languages~$(\mathcal{M}_i)_{i\in\naturals}$ that satisfies:
%if there exists a partition $(\mathcal{M}_i)_{i\in I}$ of $VA^*$, with an index set %$I=\{0,1,\ldots,n\}$ or $I=\naturals$, such that:    
\vspace{-.75\topsep}
\begin{itemize}\itemsep=.25\parskip\parskip=.5\parskip
    \item for all $i\in I$, every move from words in $\mathcal{M}_i$ leads to a word outside of $\mathcal{M}_i$ (stability property),
    \item for every~$i\in\naturals$, every word~$u \in \mathcal{M}_{i}$, and every~$j<i$,   there exists a move from~$u$ leading to a word in $\mathcal{M}_j$ (absorption property).
\end{itemize}
%then, for all $i\in I$, we have $\mathcal{M}_i=\mathcal{L}_i$.
\end{lemma}

In addition to octal games, some other well-known games have also been considered in the context of rewrite games. It is for example the case of {\em Peg-solitaire} \cite{MooreEppstein:2002, Ravikumar:2004}, where $R$ is of the form $\{\move{\tt{aab}}{\tt{bba}},\allowbreak \move{\tt{baa}}{\tt{aab}}\}$.
In {\em Peg-solitaire}, it has been proved that on one dimensional boards, the set of solvable configurations forms a regular language. In the 2-player version of the game, called {\sc duotaire}, where series of hops can be done in a single move, neither the $\P$ nor the $\N$ positions form a regular nor even a context-free language. Another example of a combinatorial game seen as a rewrite game is the game {\sc clobber} \cite{clobber}, played over a $3$-letter alphabet $\{\ta,\tb,\emptyset\}$ with 
$R=\{
\move{\tt{ab\emptyset}}{\tt{\emptyset\emptyset\ta}},
\move{\tt{ba\emptyset}}{\tt{\emptyset\emptyset\tb}}
\}$.

\medskip

In the following, we assume that the reader is familiar with basic results about formal languages and combinatorial games. 
We refer the reader respectively to~\cite{Hopcroft} and~\cite{Siegel}
for a general reference on these topics.
For any given letter $\ta$ and word~$w$, we let~$|w|_\ta$  denote the number of $\ta$ occurring in the word $w$.
We denote by~$\varepsilon$ the empty word.

\subsection*{Taking-and-merging games}

In most of this article, we consider a family of rewrite games over a two-letter alphabet, say $\{\ta,\tb\}$, where any reduction rule of $R$ is either of the form $\emove{\ta^k}$ or $\emove{\tb^k}$ for some $k$. In a certain way, this family allows us to model a new kind of pile games, where taking moves are combined with merging ones. For example, by following Waldmann's description of octal games with a rewrite system, playing $\emove{\tb}$ from the word $\ta\tb\ta^5$ leads to $\ta^6$ and can be seen as a merging of the piles $\ta$ and $\ta^5$.\\

From now on and for the sake of notation, we will omit the reduction to $\varepsilon$ in the description of the rewrite system. In other words, the games considered here will be denoted by a set
\begin{equation*}
\{\ta^{k_1},\ta^{k_2},\ldots,\ta^{k_n},\tb^{\ell_1},\tb^{\ell_2},\ldots,\tb^{\ell_m}\}
\end{equation*}
where the $k_i$ and $\ell_i$ are positive integers.\\

We now consider a first example of such a taking-and-merging game. Using a convenient invariant (denoted by $S$) is a strategy that will appear in several proofs encountered in this paper.
\begin{example}
Let us consider the game $G=\{\ta^2, \tb\}$. We claim that the DFA (deterministic finite automaton) depicted in Figure~\ref{fig:a2b} computes the Grundy function of $G$: consider a word $w$ and start reading it from the initial state marked with an incoming arrow. Follow transitions reading the word letter by letter from left to right and look at the state reached when reading the last letter of $w$. The states $(0.0)$ and $(0.1)$ correspond to the words of Grundy value $0$, and the states $(1.2)$ and $(1.3)$ to those of Grundy value $1$. First, note that this is true for the two final positions $\varepsilon$ and $\ta$. To prove this result, we define the following quantity for a given word $u$.
\begin{equation*}
    S(u)=(|u|_a-2|u|_b)\bmod{4}
\end{equation*}

  \begin{figure}[htbp]
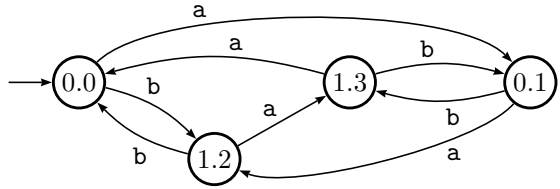

        \centering
\VCDraw{%
        \begin{VCPicture}{(-1,-3)(11,2)}
\LargeState
% states
 \State[0.0]{(0,0)}{0}
 \State[1.3]{(6,0)}{3}
 \State[0.1]{(10,0)}{1}
 \State[1.2]{(3,-1.7)}{2}
% initial--final
\Initial[w]{0}
%\Final[s]{4}
%\Final[s]{5}
% transitions
\ArcR{3}{0}{\ta}
\VArcL[.3]{arcangle=45,ncurv=.4}{0}{1}{\ta}
\ArcL{1}{3}{\tb}
\ArcL{3}{1}{\tb}
\ArcL{0}{2}{\tb}
\ArcL{2}{0}{\tb}
\EdgeL{2}{3}{\ta}
\VArcL[.3]{arcangle=35,ncurv=.4}{1}{2}{\ta}
\end{VCPicture}
}
\caption{A DFA computing the Grundy function of the game $\{\ta^2, \tb\}$.}\label{fig:a2b}
    \end{figure}

One can first observe that for all $i=0,\ldots,3$, every word $u$ recognized by the state $(X.i)$ (for $X\in\{0,1\}$) satisfies $S(u)=i$. To check this property, it suffices to consider each transition of the DFA and verify that $S(u)$ changes accordingly. For example, reading a letter $\ta$ from the state $(1.2)$ increases by $1$ the value of $S(u)$, leading to the state $(1.3)$, while reading a letter $\tb$ decreases by $2$ the value and leads to the state $(0.0)$. Then, in order to prove that the DFA computes the Grundy values, by Lemma~\ref{lem:Grundy}, it suffices to show that any move from a word recognized by a state $(0.X)$ (for $X\in\{0,1\}$) leads to a word recognized by a state $(1.Y)$ (for some $Y\in\{2,3\}$), and that any move from a word recognized by a state $(1.Y)$ leads to some $(0.X)$. These two properties can be easily checked by using the invariant $S(u)$:
\begin{itemize}
\item By definition of $S$, any move $\emove{\ta^2}$ from a word $u$ such that $S(u)=0,1$ leads to a word $u'$ having $S(u')=2,3$, and conversely.
\item Any move $\emove{\tb}$ satisfies the same property, as $S$ is modified by $2\bmod{4}$. 
\end{itemize}
\end{example}

In view of such an example and according to Guy's conjecture, it is natural to wonder whether the regularity of the languages $\mathcal{L}_i$ would hold in the context of taking-and-merging games. In Section 2, we will give a negative answer to this question, for games where both reductions $\emove{\ta}$ and  $\emove{\tb}$ are forbidden. In addition, a proof of context-freeness is given for simple instances of such games. In Section 3, we prove the regularity of several taking-and-merging games. In particular, we exhibit DFAs computing their Grundy functions. Section 4 deals with a discussion about a result of Waldmann about the correlation between the regularity of $\mathcal{L}_0$, the other $\mathcal{L}_i$, and the number of Grundy values. 
The last section explains why we restricted our study to taking-and-merging games: in  the slightly more general settings of strongly-terminating rewrite games (i.e., where each move strictly decreases the length of the position), some problems become undecidable.
Indeed, we show that then it is undecidable whether there exists a winning position in a given regular language $L$ of starting positions.

\section{Not all games lead to regular languages}

Our first result shows that Guy's conjecture does not hold for taking-and-merging games. More precisely, it states that, considering any taking-and-merging game that excludes both reductions $\emove{\ta}$ and  $\emove{\tb}$,
the set of $\P$-positions is not a regular language.

\subsection{Games $\{\ta^{k_1},\ldots,\ta^{k_n},\tb^{\ell_1},\ldots,\tb^{\ell_m}\}$ with $k_1>1$ and $\ell_1>1$}

\begin{theorem}\label{the:notregular}
Let $G$ be the taking-and-merging game  $\{\ta^{k_1},\ldots,\ta^{k_n},\tb^{\ell_1},\ldots,\tb^{\ell_m}\}$, with $k_1\leq k_2 \leq \ldots\leq k_n$ and $\ell_1 \leq \ell_2 \leq \ldots\leq \ell_m$. If $k_1>1$ and $\ell_1>1$, then the language of the $\P$-position of $G$ is not regular.
\end{theorem}

\begin{proof}
Let us show that the intersection of the set of $\P$-positions of $G$ with the regular language $L$, defined below, is not a regular language.
\begin{equation*}
L=\tb^{\ell_1-1}(\ta\tb^{\ell_1-1})^*(\tb\ta^{k_1-1})^*
\end{equation*}
More precisely, we prove by induction that the word $u_{i,j}=\tt{b}^{\ell_1-1}(\tt{ab}^{\ell_1-1})^i(\tt{ba}^{k_1-1})^j$ is a $\P$-position if and only if $i\geq j$.

If $i=0$ and $j>0$, then there is only one valid move from position $u_{i,j}$ and it leads to position $f$, below.
\begin{equation*}
    u_{0,j} = \tb^{\ell_1-1}(\tb\ta^{k_1-1})^j~\movearrow~f=\ta^{k_1-1}(\tb\ta^{k_1-1})^{j-1}
\end{equation*}
It may be verified that~$f$ is a final position, hence that~$u_{0,j}$ is a $\N$-position, for every~$j>0$.
On the other hand, for every $i\geq 0$ then $u_{i,0}$ is a final position, hence a $\P$-position. 
In other words, the claim is true if $i=0$ or $j=0$.

Now, assume that $i>0$ and $j>0$.
In that case, we denote by~$v$ the following word.
\begin{align*}
v=\tb^{\ell_1-1}(\ta\tb^{\ell_1-1})^{i-1}\ta^{k_1}(\tb\ta^{k_1-1})^{j-1}
%&&
%w=\tb^{\ell_1-1}(\ta\tb^{\ell_1-1})^{i-1}(\tb\ta^{k_1-1})^{j-1}
\end{align*}
It may be verified that only one move is valid from~$u_{i,j}$,
and that it leads to~$v$. Similarly,
the only move from~$v$ leads to~$u_{i-1,j-1}$.
Therefore, words $u_{i,j}$ and $u_{i-1,j-1}$ have the same outcome and by induction hypothesis $u_{i-1,j-1}$ is a $\P$-position if and only if $i-1 \geq j-1$, which concludes the induction.
%
%Finally, the intersection of the regular language $L$ and the $\P$-positions is not regular and thus the set of $\P$-positions is not a regular set.
\end{proof}

\subsection{Context-freenes for $\{\ta^{k},\tb^{\ell}\}$}\label{sec:22}

We have seen with Theorem~\ref{the:notregular} that the language made of $\P$-positions is, in general, not regular. Nevertheless, when limited to a rewrite game with two reductions, we get the following result.

\begin{theorem}\label{th:tak,tbl}
Let $k,\ell$ be positive integers. The taking-and-merging game $\{\ta^{k},\tb^{\ell}\}$ has only two Grundy values and the corresponding languages $\mathcal{L}_0$ and $\mathcal{L}_1$ are context-free.
\end{theorem}

\begin{proof}
%
%
%Let $u,v,w$ be three words such that $u\longrightarrow_R^* v\longrightarrow_R^* w$. 
%Consequently, if a sequence of reductions from $u$ to $v$ has length $\alpha$ and a sequence of reductions from $v$ to $w$ has length $\beta$, then every sequence of reductions from $u$ to $w$ has length $\alpha+\beta$. 
%
The rewrite system $\{\emove{\ta^{k}}, \emove{\tb^{\ell}}\}$ is \emph{weakly confluent},
that is, if~$\move{u}{v_1}$ and~$\move{u}{v_2}$, then there exists a~$w$
such that~$v_1\movearrow^* w$ and~$v_2\movearrow^* w$ (in our case,~$w$ can be reached in at most one step).
Since moreover, this rewriting system is \emph{terminating} (i.e., there is no infinite rewriting chain), Newman's Lemma \cite{Terese:2003} yields that the rewriting system is \emph{confluent} or, stated otherwise, from any position~$u$ can be reached a unique final position.

Let $u$ be a word and~$w$ be the unique final position reachable from~$u$.
If $|u|_\ta=n$, $|u|_\tb=m$, there exists $\alpha,\beta\ge 0$ such that $|w|_\ta=n-\alpha\, k$ and $|w|_\tb=m-\beta\, \ell$. This means that the reduction $\emove{\ta^k}$ (resp.\@ $\emove{\tb^\ell}$) has been applied $\alpha$ (resp.\@ $\beta$) times in a sequence of  $\alpha+\beta$ reductions. 
Hence, playing the game starting from~$u$ necessarily consists of $\alpha+\beta$ moves. Consequently~$u$ is a $\P$-position (resp.~a $\N$-position) if and only if $\alpha+\beta$ is even (resp.~odd)

To compute the Grundy value of a word, one just has to apply all the possible reductions in any order and count the parity of the number of applied reductions. This can be computed by a push-down automata: reading the word from left to right, each time there are $k$ consecutive letters $\ta$ or $\ell$ consecutive letters $\tb$, a reduction is simulated and the parity changed. Let us define more formally this push-down automata. It has with three states: $0,1$ and an initial state $q_0$. 
The stack alphabet is 
$$\{ (\ta,1),\ldots,(\ta,k-1),(\tb,1),\ldots,(\tb,\ell-1),\perp\}$$
where $\perp$ is a special symbol to represent the bottom of the stack. 
Transitions are of the form 
$$(i,x,y,z,j)$$
where $i,j\in\{0,1\}$ are states, $x$ is the symbol read by the automata, $y$ is the symbol that is popped from the top of the stack, $z$ is the word that is then pushed on the stack (with the usual convention that the leftmost symbol is on top of the stack). 

First, there is a unique transition leaving the initial states; it initializes the stack with the bottom symbol~$\bot$ without reading any letter from the input:
\begin{equation*}
(q_0,\varepsilon,\varepsilon,\perp,0)~.
\end{equation*}
Second, the transition table for states $q\in\{0,1\}$ is given in Table~\ref{tab:transitions}.

\begin{table}[ht]\arraycolsep=5pt\def\arraystretch{1.3}\centering
\begin{tabular}{cccccl}
\text{source}& \text{input} &  \text{popped}& \text{pushed} & \text{target} \\[-0.5em]
\text{state}& \text{letter} & \text{symbol}& \text{symbols}& \text{state} \\
\cline{1-5}
($q$, &$\ta$, &$\perp$,        &$(\ta,1)\perp$,   &$q$)   & \\
($q$, &$\tb$, &$\perp$,        &$(\tb,1)\perp$,   &$q$)   & \\
($q$, &$\ta$, &$(\tb,j)$,      &$(\ta,1)(\tb,j)$, &$q$)   & for each $j<\ell$\\
($q$, &$\tb$, &$(\ta,j)$,      &$(\tb,1)(\ta,j)$, &$q$)   & for each $j<k$\\
($q$, &$\ta$, &$(\ta,i)$,      &$(\ta,i+1)$,      &$q$)   & if $i<k-1$\\
($q$, &$\tb$, &$(\tb,i)$,      &$(\tb,i+1)$,      &$q$)   & if $i<\ell-1$\\
($q$, &$\ta$, &$(\ta,k-1)$,    &$\varepsilon$,    &$1-q$) & \\
($q$, &$\tb$, &$(\tb,\ell-1)$, &$\varepsilon$,    &$1-q$) &
\end{tabular}
\caption{Transition table for states~$q\in\{0,1\}$}
\label{tab:transitions}
\end{table}

For each of these transitions, observe that a symbol has to be popped from the stack. 
We store on the stack the blocks of letters that are were read but not consumed: note that symbol $(\ta,5)$ means a block of five~$\ta$'s.
If a block of $k$ contiguous~$\ta$'s is found,
that is if we read~$\ta$ from the input and that $(\ta,k-1)$ is the symbol on top of the stack, we apply~$\emove{\ta^{\ell}}$, effectively popping $(\ta,k-1)$ from the stack. Moreover, the automaton goes into the other state (from~$0$ to~$1$ or~$1$ to~$0$).
A similar transition is taken when a block of~$\ell$ contiguous~$\tb$'s is found.
In all other cases, the stack is simply updated without changing the state.
When the input is entirely read, the state of the automaton is the parity of the number of reductions that have been applied. 
We disregard the final content of the stack; it is the final position of the game.
\end{proof}

\begin{remark}\label{rk:akb}
In the above result, when $k$ or $\ell$ is equal to $1$, the two languages $\mathcal{L}_0$ and $\mathcal{L}_1$ are regular. Indeed, the stack is not needed in that case. 

Assume that $k>1$ and $\ell=1$. Since the order of the moves does not matter, we may assume that all the moves $\emove{\tb}$ are played first. 
Then the word contains only letters $\ta$ and the rule $\emove{\ta^k}$ is played until a position $\ta^i$ with $i<k$ is reached. Thus, the number of moves from a starting position $u$ is $|u|_{\tb}+\left\lfloor \frac{|u|_{\ta}}{k}\right\rfloor$ and the Grundy value is the parity of this number. This can easily be computed by a DFA. In Figure~\ref{fig:dfa31}, we have represented the DFA for the game $\{\ta^3,\tb\}$. (The integers in the states are the Grundy values.)
\begin{figure}[ht!b]
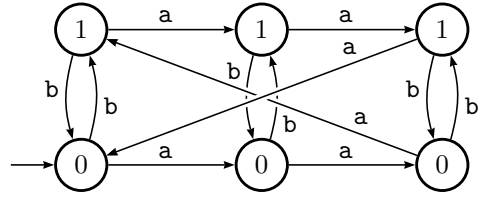

\begin{center}
\VCDraw{%
\begin{VCPicture}{(-1,1)(9,6)}
\LargeState
% states
 \State[0]{(0,2)}{0}
 \State[0]{(4,2)}{1}
 \State[0]{(8,2)}{2}
 \State[1]{(0,5)}{3}
 \State[1]{(4,5)}{4}
 \State[1]{(8,5)}{5}
% initial--final
\Initial[w]{0}
% transitions 
\EdgeL{0}{1}{\ta}
\EdgeL{1}{2}{\ta}
\EdgeL{3}{4}{\ta}
\EdgeL{4}{5}{\ta}
\ArcR{0}{3}{\tb}
\ArcR{3}{0}{\tb}
\ArcR[.2]{1}{4}{\tb}
\ArcR[.2]{4}{1}{\tb}
\ArcR{2}{5}{\tb}
\ArcR{5}{2}{\tb}
\EdgeBorder
\EdgeR[.2]{2}{3}{\ta}
\EdgeR[.2]{5}{0}{\ta}
\EdgeBorderOff
\end{VCPicture}
}
\end{center}
\caption{The DFA computing the Grundy values of $\{\ta^3,\tb\}$.}\label{fig:dfa31}
\end{figure}

\end{remark}

\begin{remark}
    The proof of Theorem~\ref{th:tak,tbl} generalizes to any $n$-letter game of the form~$\{a_1^{k_1},\ldots, a_n^{k_n}\}$.
\end{remark}

\section{Regularity of some games}

In this section, we prove the regularity of some games of the form \begin{equation*}
G=\{\ta^{k_1},\ta^{k_2},...,\ta^{k_n},\tb\}~.
\end{equation*}
If the game has only two rules, $\{\ta^{k_1},\tb\}$, the game is trivial: there are only two Grundy values and the two corresponding languages $\mathcal L_0$ and $\mathcal L_1$ are regular (see Remark \ref{rk:akb}).
In the following, we consider games with at least three rules.

\subsection{The game $\{\ta,\ta^{2k+1},\tb\}$}

In the game $\{\ta,\ta^{2k+1},\tb\}$ the only irreducible word is $\varepsilon$ (since $\emove{\ta}$ and $\emove{\tb}$ are moves), and all words $w\in A^*$ can be reduced to it. Let $w$ be a word. We need $|w|_\tb$ reductions of the form $\emove{\tb}$ to get rid of the $\tb$'s. To get rid of all the $\ta$'s, since the reduction rules all involve an odd number of $\ta$, the number of reductions to apply to eliminate the $\ta$'s has the same parity as $|w|_\ta$. Hence, the number of reductions to apply to a word $w$ to obtain $\varepsilon$ is even if and only if $|w|_\ta+|w|_\tb$ is even. 
Let us partition $A^*$ into two sets~$\mathcal{M}_0$ and~$\mathcal{M}_1$; a word~$w$ belongs to~$\mathcal{M}_0$ if~$|w|_\ta+|w|_\tb$ is even and to~$\mathcal{M}_1$ if it is odd. 
Lemma~\ref{lem:Grundy} then yields that~$\mathcal{M}_0$ is the set of $\P$-positions and that~$\mathcal{M}_1$ is the set of $\N$-position.
It can be easily shown that these two languages are regular.

\begin{remark}
The same argument extends to each game whose set of rewriting rules contains $\emove{\ta}$, $\emove{\tb}$ and any number of rules of the form $\emove{\ta^{2k+1}}$ and $\emove{\tb^{2\ell+1}}$.
\end{remark}

\subsection{The game $\{\ta,\ta^2,\tb\}$}

In this section, we prove that for the game $\{\ta,\ta^2,\tb\}$, the language $\mathcal L_i$ of words of Grundy value $i$ is regular for any Grundy value $i$ and we explicitly give a DFA that computes the Grundy values.

Every word in $A^*$ can be uniquely written as
$$w=\ta^{i_0}\tb\ta^{i_1}\tb\cdots \tb\ta^{i_k}$$
where $k\ge 0$ and $i_0,\ldots,i_k\ge 0$. 
With every word $w$ is thus associated a tuple $(i_0,\ldots,i_k)$, with~$k=|w|_b$;
this association is one-to-one.
For $j\in \{0,...,k\}$, let $i_j':=i_j\bmod{3}$. 
For $r \in \{1,2\}$, let ${\alpha_r=\#\{j \mid i_j'=r\}}$ be the number of blocks of $\ta$'s of size $r$ (modulo $3$).
Finally, we define for every word $w$ the quantity 
\begin{equation*}
  S(w)=2k+2\alpha_1+\alpha_2 \bmod 4~.
\end{equation*}
As an example, the word $w=\ta^5\tb^2\ta\tb\ta^2$ has $k=3$, $(i_0,i_1,i_2,i_3)=(5,0,1,2)$ thus $\alpha_1=1$, $\alpha_2=2$ and $S(w)=2$.

\begin{lemma}\label{lem:gtos}
Let $w\in A^*$, the Grundy value of $w$ in the game $\{\ta,\ta^2,\tb\}$ is entirely determined by $S(w)$, i.e., 
$$\g(w)=\left\{
\begin{array}{ll}
0, & \text{if } S(w)=0;\\
1, & \text{if } S(w)=2;\\
2, & \text{if } S(w)=1;\\
3, & \text{if } S(w)=3.\\
\end{array}\right.$$
% \begin{center}
% \begin{tabular}{|c|c|c|c|c|}
% \hline
% $S(w)$ & 0 & 1 & 2 & 3 \\ \hline 
% $\g(w)$ & 0 & 2 & 1 & 3 \\ \hline
% \end{tabular}
% \end{center}
\end{lemma}

\begin{proof}
The proofs consists in showing that the conditions of Lemma~\ref{lem:Grundy} are met by the following family of languages: $\mathcal M_0=S^{-1}(0)$, $\mathcal M_1=S^{-1}(2)$, $\mathcal M_2=S^{-1}(1)$, $\mathcal M_3=S^{-1}(3)$ and $\mathcal{M}_i=\emptyset$, for each~$i>3$.

First, let us show that playing any move changes the value of $S(w)$ modulo~4. Let $w\in A^*$ and consider each rule. 

\begin{itemize}
  \item If the rule $\emove{\ta}$ is played on a block $\ta^{i_r}$, then $S(w)$ decreases by 2 if $i_r'=1$ and increases by 1 if $i_r'\in \{0,2\}$. 
  \item If the rule $\emove{\ta^2}$ is played, on a block $\ta^{i_r}$, then $S(w)$ increases by 2 if $i_r'=0$,  by 1 if $i_r'=1$ and  decreases by 1 if $i_r'=2$. 
  \item Finally, assume that the rule $\emove{\tb}$ is played. Let $\ta^{i_m}$ and $\ta^{i_{m+1}}$ be the two blocks around the $\tb$ that will be removed. Table \ref{tab:rulebaa2} gives, for every value of $i_m'$ and $i_{m+1}'$, the variation of $S(w)$ modulo $4$.

\begin{table} 
   \begin{center}
\begin{tabular}{|c|c|c|c|}
\hline
\diagbox{$i'_m$}{$i'_{m+1}$} & 0 & 1 & 2 \\ \hline 
0 & -2 & -2 & -2\\ \hline
1 & -2 & -1 & -1 \\ \hline
2 & -2 & -1 & -2 \\ \hline
\end{tabular}
\end{center}
\caption{\label{tab:rulebaa2} Variation of $S(w)$ when the rule $\emove{\tb}$ 
is applied to a $\tb$ between two blocks of $\ta$'s respectively of length $i'_m$ and $i'_{m_+1}$ modulo 4.}
\end{table}
 
 As an example, consider the case $i_m'=i_{m+1}'=1$. Then one $\tb$ and two blocks of size $1$ (modulo $3$) are lost, decreasing the value of $S(w)$ by $6$, but we obtain a new block of size $2$. Thus the total value $S(w)$ decreases by $5$ which is congruent to $1$ modulo $4$. Note that if $i_m=0$ (respectively $i_{m+1}'=0$), then the number of blocks of $\ta$ of size $1$ and $2$ do not change modulo 3 and only one $\tb$ is removed, decreasing by $2$ the value $S(w)$.
\end{itemize}

Now, let us prove that $S(w)>0$, there is a move to $w'$ with $S(w')=0$.
If $S(w)$ is odd, $\alpha_2$ is also odd and in particular, there must be a block $\ta^{i_m}$ with $i_m'=2$.
Then playing $\emove{\ta^2}$ if $S(w)=1$ or $\emove{\ta}$ if $S(w)=3$ on this block leads to a word $w'$ with $S(w')=0$.
Thus assume that $S(w)=2$. If there is a block $\ta^{i_m}$ with $i_m'=1$ then playing the rule $\emove{\ta}$ on this block decreases $S(w)$ by $2$. Otherwise, there must be at least one $\tb$. Then, using Table \ref{tab:rulebaa2}, removing any $\tb$ decreases $S(w)$ by $2$ since the blocks around $\tb$ have size $0$ or $2$ modulo $3$.

If $S(w)\in \{1 ,3\}$, then there is a move to a word $w'$ with $S(w')=2$.
Indeed, as before, there must be a block $\ta^{i_m}$ with $i_m'=2$.
Then playing $\emove{\ta}$ if $S(w)=1$ or $\emove{\ta^2}$ if $S(w)=3$ on this block leads to a word $w'$ with $S(w')=2$.

Finally, if $S(w)=3$, there is a move to a word $w'$ with $S(w')=1$. We do the same reasoning than before to find a move from $S(w)=2$ to $S(w')=0$. If there is a block of size $1$ or a $\tb$ next to a block of size $0$, we remove the block of size $1$ or $\tb$. If not, we remove any $\tb$ between two blocks of size $2$.

Hence, the conditions of Lemma~\ref{lem:Grundy} are indeed met by the following family of languages: $\mathcal M_0=S^{-1}(0)$, $\mathcal M_1=S^{-1}(2)$, $\mathcal M_2=S^{-1}(1)$, $\mathcal M_3=S^{-1}(3)$ and $\mathcal{M}_i=\emptyset$, for each~$i>3$.
\end{proof}

%ou corrolaire ?

\begin{theorem}
The Grundy values of the game $\{\ta,\ta^2,\tb\}$ can be computed by a DFA.
\end{theorem}

\begin{proof}
By Lemma \ref{lem:gtos}, we just need to compute the value $S(w)$. This is done by the automaton depicted in Figure \ref{fig:aut121}.
There are 12 states. A state is denoted by $(s.i)$ where $s$ is the value $S(w)$ and $i$ is the size modulo 3 of the last block of $\ta$ of $w$. Reading $\tb$ from a state $(s.i)$ leads to state $(s-2.0)$ (values are taken modulo 4 for $s$ and modulo 3 for $i$). Reading $\ta$ from a state $(s.i)$ leads to state $(s'.(i+1))$ with $s'=s+2$ if $i=0$, $s'=s-1$ if $i\in \{1,2\}$.

\end{proof}

\begin{figure}[htbp]
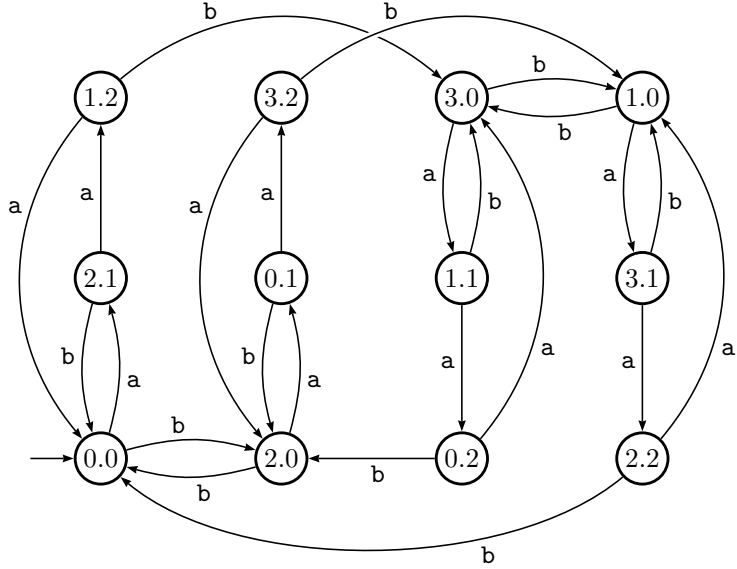

        \centering
\VCDraw{%
        \begin{VCPicture}{(-1,-3)(13,10)}
\LargeState
% states
 \State[0.0]{(0,0)}{0}
 \State[2.1]{(0,4)}{9}
 \State[1.2]{(0,8)}{6}
 \State[2.0]{(4,0)}{8}
 \State[0.1]{(4,4)}{1}
 \State[3.2]{(4,8)}{14}
 \State[0.2]{(8,0)}{2}
 \State[1.1]{(8,4)}{5}
 \State[3.0]{(8,8)}{12}
 \State[2.2]{(12,0)}{10}
 \State[3.1]{(12,4)}{13}
 \State[1.0]{(12,8)}{4}
 
% initial--final
\Initial[w]{0}
%\Final[s]{4}
%\Final[s]{5}
% transitions
%\ArcR{3}{0}{\ta}
%\VArcL[.3]{arcangle=35,ncurv=.2}{0}{1}{\ta}
%\EdgeL{2}{3}{\ta}
\ArcR{0}{9}{\ta}
\ArcR{9}{0}{\tb}
\EdgeL{9}{6}{\ta}
\VArcR[.3]{arcangle=-40,ncurv=.8}{6}{0}{\ta}
\ArcR{8}{1}{\ta}
\ArcR{1}{8}{\tb}
\EdgeL{1}{14}{\ta}
\VArcR[.3]{arcangle=-40,ncurv=.8}{14}{8}{\ta}
\ArcL{0}{8}{\tb}
\ArcL{8}{0}{\tb}
\VArcR[.3]{arcangle=-40,ncurv=.8}{2}{12}{\ta}
\EdgeR{5}{2}{\ta}
\ArcR{5}{12}{\tb}
\ArcR{12}{5}{\ta}
\VArcR[.3]{arcangle=-40,ncurv=.8}{10}{4}{\ta}
\EdgeR{13}{10}{\ta}
\ArcR{13}{4}{\tb}
\ArcR{4}{13}{\ta}
\ArcL{4}{12}{\tb}
\ArcL{12}{4}{\tb}
\EdgeL{2}{8}{\tb}
\VArcL[.3]{arcangle=40,ncurv=.6}{10}{0}{\tb}
\VArcL[.3]{arcangle=40,ncurv=.8}{6}{12}{\tb}
\EdgeBorder
\VArcL[.3]{arcangle=40,ncurv=.8}{14}{4}{\tb}
\end{VCPicture}
}
\caption{A DFA for the game $\{\ta,\ta^2, \tb\}$.}\label{fig:aut121}
    \end{figure}

What could happen if we just replace the rule $\emove{\ta^2}$ by $\emove{\ta^4}$? Surprisingly, we did not find an automaton for the game $\{\ta,\ta^4, \tb\}$ even though the Grundy values of this game seem to be bounded as suggested by computer experiments. For words of length at most 20, the Grundy function is bounded by~$3$. This leads to the following open question.

\begin{question}
Are the Grundy values of the game $\{\ta,\ta^4, \tb\}$ bounded? Are the corresponding sets regular?
\end{question}

\subsection{The game $\{\ta, \ta^2, \ta^3, \tb\}$}

We now prove that for the game $\{\ta,\ta^2,\ta^3,\tb\}$, the corresponding sets $\mathcal L_i$ are again regular and give a DFA that computes the Grundy values. 
As before, every word in $A^*$ can be written as
$$w=\ta^{i_0}\tb\ta^{i_1}\tb\cdots \tb\ta^{i_k}$$
where $k=|w|_{\tb}\ge 0$ and $i_0,\ldots,i_k\ge 0$. 
We now write $i_j':=i_j\bmod{4}$. 
For $r \in \{1,2,3\}$, let $\alpha_r=\#\{j \mid i_j'=r\}$ be the number of blocks of size $r$ modulo 4.
Finally, we define for any word the triplet of $\{0,1\}^3$.
$$S(w)=(k+\alpha_1 \bmod 2,\alpha_2 \bmod 2,\alpha_3 \bmod 2).$$

For convenience reasons, we will denote the triplet $S(w)=(x,y,z)$ by the word $xyz$. 
As an example, the word $w=\ta^5\tb\ta^3\tb\tb\ta^2\tb\ta$ has $k=4$, $\alpha_1=2$, $\alpha_2=\alpha_3=1$ and thus $S(w)=011$. As before, the value $S(w)$ is enough to compute the Grundy values.

\begin{lemma}\label{lem:a123b}
Let $w\in A^*$, the Grundy value of $w$ in the game $\{\ta,\ta^2,\ta^3,\tb\}$ is determined by $S(w)$, i.e., 
$$\g(w)=\left\{
\begin{array}{ll}
0, & \text{if } S(w)\in\{000, 111\};\\
1, & \text{if } S(w)\in\{011, 100\};\\
2, & \text{if } S(w)\in\{010, 101\};\\
3, & \text{if } S(w)\in\{001, 110\}.\\
\end{array}\right.$$
% \begin{center}
% \begin{tabular}{|c|c|c|c|c|}
% \hline
% $\g(w)$ & 0 & 1 & 2 & 3 \\ \hline 
% $S(w)$ & 000 & 011& 010& 001 \\
%  & 111 & 100 & 101 & 110 \\\hline
% \end{tabular}
% \end{center}
\end{lemma}

Note that the values $S(w)$ are paired with their complement

\begin{proof}
We denote by $\mathcal M_i$, $i\in \{0,1,2,3\}$ the potential candidates for $\mathcal L_i$, that are:
\begin{itemize}
\item $\mathcal M_0 = \{ w \in A^* | S(w) \in \{000,111\}\}$;
\item $\mathcal M_1 = \{ w \in A^*  | S(w) \in \{011,100\}\}$;
\item $\mathcal M_2 = \{ w \in A^*  | S(w) \in \{010,101\}\}$;
\item $\mathcal M_3 = \{ w \in A^*  | S(w) \in \{001,110\}\}$.
\end{itemize}

We aim to prove that $\mathcal L_i=\mathcal M_i$. We first list the evolution of $S(w)$ depending on the rule that is played.

\begin{itemize}
  \item The rule $\emove{\ta^k}$ is played on a block $\ta^{i_r}$. Table \ref{tab:ruleai123} gives the vector (in a compact form) that is applied to $S(w)$ in function of the values of $i_r'$ and the rule $\ta^k$. 
  As an example, consider the case $i_r'=2$ and the rule $\emove{\ta^3}$. One block of size $2$ is replaced by a block of size $2-3=3 \bmod {4}$. Thus the vector applied to $S(w)$ is $(0,1,1)$ (values are taken modulo 2)..

\begin{table}[ht]   
\begin{center}
\begin{tabular}{|c|c|c|c|c|}
\hline
\diagbox{$k$}{$i_r'$} & 0 & 1 & 2 & 3\\ \hline 
1 & 001 & 100 & 110 & 011\\ \hline
2 & 010 & 101 & 010 & 101 \\ \hline
3 & 100 & 110 & 011 & 001 \\ \hline
\end{tabular}
\end{center}
\caption{Variation of $S(w)$ with the rules $\emove{\ta^i}$ on a block $\ta^{i_r}$.}\label{tab:ruleai123} 
\end{table}
 
    \item The rule $\emove{\tb}$ is played. Let $\ta^{i_m}$ and $\ta^{i_{m+1}}$ be the two blocks around the $\tb$ that is removed. Table \ref{tab:ruleb123} gives the vector applied to $S(w)$ depending on the values of $i_m'$ and $i_{m+1}'$.
\end{itemize}

\begin{table}[ht]
\centering
\begin{tabular}{|c|c|c|c|c|}
\hline
\diagbox{$i_m'$}{$i_{m+1}'$} & 0 & 1 & 2 & 3 \\ \hline 
0 & 100 & 100 & 100 & 100\\ \hline
1 & 100 & 110 & 011 & 001 \\ \hline
2 & 100 & 011 & 100 & 011 \\ \hline
3 & 100 & 001 & 011 &110 \\\hline
\end{tabular}
\caption{ Variation of $S(w)$ when the rule $\emove{\tb}$ 
is applied to a $\tb$ between two blocks of $\ta$'s respectively of length $i'_m$ and $i'_{m_+1}$ modulo 4.}
\label{tab:ruleb123}
\end{table}

In both cases, there is no variation with vector $000$ or $111$ which proves that all the sets $\mathcal M_i$ are stable. We now prove that for any word in $\mathcal M_i$ there is a move to a word in $\mathcal M_j$ if $i>j$.

First note that, except if $w$ contains only $\ta$'s and an even number of them, it is always possible to change $S(w)$ by either vector $100$ or $011$.
Indeed, consider such a word $w$. If there is a block of $\ta$'s of size $1$ or $3$, then playing $\emove{\ta}$ to any $\ta$ of this block changes the value of $S(w)$ by $100$ or $011$. 
Otherwise, there are only blocks of size $0$ or $2$ (modulo 4), and necessarily one $\tb$. Then playing $\emove{\tb}$ to any $\tb$ changes the value of $S(w)$ by $100$ or $011$ (according to Table \ref{tab:ruleb123}). This remark implies that there is always a move from a word in $\mathcal M_1$ to $\mathcal M_0$ and from a word in $\mathcal M_3$ to $\mathcal M_2$.

%Aline : j'en suis là 

Second, we prove that if $w$ belongs to~$\mathcal{M}_2\cup \mathcal{M}_3$, it is always possible to change $S(w)$ by either vector $001$ or vector $110$.  By definition of $\mathcal M_2$ and $\mathcal M_3$, there is always in $w$ either a block of size $2$ or a block of size $3$ modulo 4. Then, using Table~\ref{tab:ruleai123}, playing $\emove{\ta}$ (in the first case) or $\emove{\ta^3}$ (in the second case) changes the value of $S(w)$ with vector $110$ (in the first case) or $001$ (in the second case). This implies that there is always a move from a word in $\mathcal M_2$ to $\mathcal M_1$ and from a word in $\mathcal M_3$ to $\mathcal M_0$.

Third we prove that if $w$ belong to~$\mathcal{M}_2\cup \mathcal{M}_3$, it is always possible to change $S(w)$ by either vector $010$ or vector $101$. As before,  $w\in \mathcal M_2\cup \mathcal M_3$ must contain either a block of size $2$ or a block of size $3$ modulo 4. Then playing $\emove{\ta^2}$ changes the value of $S(w)$ with vector $010$ (in the first case) or $101$ (in the second case). This implies that there is always a move from a word in $\mathcal M_2$ to $\mathcal M_0$ and from a word in $\mathcal M_3$ to $\mathcal M_1$.

Lemma \ref{lem:Grundy} yields that $\mathcal M_i=\mathcal L_i$ for all $i\in \{0,1,2,3\}$.
\end{proof}

\begin{theorem}\label{th:a,aa,aaa,b}
The Grundy values of the game $\{\ta,\ta^2,\ta^3,\tb\}$ can be computed by a DFA.
\end{theorem}

\begin{proof}
We construct a DFA that computes the Grundy values of the game $\{\ta,\ta^2,\ta^3,\tb\}$.
By Lemma \ref{lem:a123b}, one just needs to compute the value of $S(w)$, which, by definition of $S(w)$, can be done by an automaton that stores the value of $k$, $\alpha_i$, $i\in \{1,2,3\}$ modulo 2 and the number modulo 4 of $\ta$ in the last block.
\end{proof}

\begin{remark}
In the proof of Theorem~\ref{th:a,aa,aaa,b}, we don't need to maintain~$S(w)$ entirely,
if all we want is the Grundy value. 
Indeed, it is enough to store the Grundy value and the parity of the last block of $\ta$.

This is due to the fact that, with a given parity for the last block of $\ta$'s, 
in order to obtain~$S(wx)$ from~$S(w)$ for some word~$w$ and letter~$x$, we apply some vector or its complement to~$S(w)$.
For instance, if $w$ ends with an odd number of $\ta$'s and that $\ta$ is read, then the vector applied to $S(w)$ is $110$ if $w$ ends with one $\ta$ and $001$ if it ends with $\ta^3$. The automaton computing the Grundy values in this way is depicted in Figure \ref{fig:auta123b}. It has eight states and the label $(g.i)$ in a state indicates that $g$ is the Grundy value and $i$ the parity of the number of $\ta$ in the last block.
As an example, from state $(3.1)$, when reading $\ta$, $S(w)$ changes by vector $110$ or $001$ and thus the Grundy value that was $3$ is now $0$ and there are now an even number of $\ta$. Hence the new state is $0.0$. When reading $\tb$, $S(w)$ changes by vector $100$ and thus the Grundy value is now $2$ and the final letter is $\tb$, thus the new state is $(2.0)$.
\end{remark}

\begin{figure}[htbp]
        \centering
\VCDraw{%
        \begin{VCPicture}{(-4.25,5.25)(15.75,12.75)}
\LargeState
% states
 \State[1.0]{(0,12)}{4}
 \State[0.1]{(3,9)}{1}
 \State[2.1]{(-3,9)}{9}
 \State[3.0]{(0,6)}{12}
 \State[2.0]{(12,6)}{8}
 \State[1.1]{(15,9)}{5}
 \State[3.1]{(9,9)}{13}
 \State[0.0]{(12,12)}{0}
% initial--final
\Initial[n]{0}
%\Final[s]{4}
%\Final[s]{5}
% transitions
\ArcL{4}{1}{\ta}
\ArcL{1}{4}{\tb}
\ArcL{1}{12}{\ta}
\ArcL{9}{4}{\ta}
\ArcL{9}{12}{\tb}
\ArcL{12}{9}{\ta}
\ArcL{0}{5}{\ta}
\ArcL{5}{0}{\tb}
\ArcL{5}{8}{\ta}
\ArcL{13}{0}{\ta}
\ArcL{13}{8}{\tb}
\ArcL{8}{13}{\ta}
\VArcL{arcangle=8,ncurv=.5}{12}{8}{\tb}
\VArcL{arcangle=8,ncurv=.5}{8}{12}{\tb}
\VArcL{arcangle=8,ncurv=.5}{0}{4}{\tb}
\VArcL{arcangle=8,ncurv=.5}{4}{0}{\tb}
\end{VCPicture}%
}
\caption{DFA for the game $\{\ta,\ta^2,\ta^3,\tb\}$.}\label{fig:auta123b}
    \end{figure}

One could hope to show a similar result for each game of the form $\{\ta,\ta^2,...,\ta^k,\tb\}$, by computing the number of $\tb$ and blocks of $\ta$ of size $1$, $2$, $3$, ..., $k-1$ modulo $2$ and finding some invariant for the Grundy values. However, this method already fails for $k=4$ since the word $\ta^2\tb\ta^2$ is a $\P$-position for this game whereas $\tb$ is not. 
We have computed the Grundy values for all words of length up to $23$ for this game and already found $14$ Grundy values:
$$(\max\{\g(u)\})_{|u|=0,1,2,\ldots}=0,1,2,3,4,5,5,6,7,7,7,7,7,8,9,9,10,11,11,12,13,13,13,14$$
This suggests that the automaton, if it exists for this game, is not as simple as it was for $k=2$ or $k=3$.

% seems false
%\begin{conjecture}[VM TODO]
%    The game~$\{\ta,\ta^ {2},\ldots,\ta^{k},\tb\}$ has~$\ell$ Grundy values that may be computed as follows.
%    Let~$w$ be a position; it is written uniquely as
%    \begin{equation*}
%      w= \ta^{i_0} \tb \ta^{i_1} \tb \cdots \tb a^{i_k}
%    \end{equation*}
%    where $k=|w|_{\tb}\ge 0$ and $i_0,\ldots,i_k\ge 0$. 
%We now write $i_j':=i_j\bmod{k+1}$. 
%For $r \in \{1,2,3\}$, let $\alpha_r=\#\{j \mid i_j'=r\}$ be the number of blocks of %size $r$ modulo 4.
%Finally, we define for any word the $k$-tuple of $\{0,1\}^k$.
%$$S(w)=(k+\alpha_1 \bmod 2,\alpha_2 \bmod 2,\ldots,\alpha_{k} \bmod 2).$$

%$S(w)$ may be seen as the least significant digit first encoding of an integer~$N(w)$ %in~$\{0, 2^k-1\}$; the Grundy value of~$w$ is~$N(w)$ if~$N(w)<2^{k-1}$,
%or~$N(\overline{w})$ otherwise (where $\overline{w}$ denotes the bitwise complement %of~$w$)
%\end{conjecture}

\section{Does a regular set of $\P$-positions imply regular sets of Grundy values?}

Given a rewrite game, deciding whether each set $\mathcal{L}_i$ forms a regular language remains an open problem in a certain number of cases. Therefore, it seems natural to know whether a positive or a negative answer can be given without considering all the sets. A first step towards this direction has been given by Waldmann, who obtained the following result \cite[Thm.~6]{Waldmann:2002}.

\begin{theorem}[Waldmann, 2002]\label{thm:l0}
For all taking-and-breaking games, if the language $\mathcal{L}_0$ is regular, then the Grundy function is bounded, and all the Grundy languages $\mathcal{L}_i$ are regular.
\end{theorem}

Hence in the case of taking-and-breaking games, the regularity of $\mathcal{L}_0$ implies the regularity of all the languages $\mathcal{L}_i$. In our different setting of taking-and-merging games, Waldmann's proof cannot be transposed easily. In addition, the situation does not seem that clear. Let us consider a particular game. 

\begin{proposition}
For the game $\{\ta,\ta^2,\tb,\tb^2\}$, the set $\mathcal{L}_0$ of $\P$-positions is regular.
\end{proposition}

\begin{proof} Consider the partition of $A^*$ into two sets 
$$P=\{ w\in A^* : |w|_\ta-|w|_\tb=0\pmod{3}\} \text{ and }
N=A^*\setminus P.$$
The set $P$ satisfies the stability property of Lemma~\ref{lem:Grundy}: take any word $w\neq\varepsilon$ in $P$ and apply one of the reductions. The resulting word $u$ is such that
$$|u|_\ta-|u|_\tb \in \big(|w|_\ta-|w|_\tb+\{-2,-1,1,2\}\big)~.$$
Hence there is no move between two words in $P$. 

The set $P$ is absorbing: take a word $w$ such that $|w|_\ta-|w|_\tb=1\bmod{3}$. If $|w|_\ta>0$, then using the reduction $\emove{\ta}$ leads to the set $P$. Otherwise, $w$ contains only $\tb$'s. Notice that it contains at least two $\tb$'s and using the reduction $\emove{\tb^2}$ leads again to the set $P$. Now take a word $w$ such that $|w|_\ta-|w|_\tb=2\bmod{3}$.  The argument is similar. If $|w|_\tb>0$, then using the reduction $\emove{\tb}$ leads to the set $P$. Otherwise, $w$ contains only $\ta$'s. Notice that it contains at least two $\ta$'s and using the reduction $\emove{\ta^2}$ leads again to the set $P$.

Hence according to Lemma~\ref{lem:Grundy}, the set $P$ is the set of the $\mathcal{P}$-positions of the game and is exactly $\mathcal{L}_0$. It is a straightforward exercise to see that $P$ is a regular language recognized by a DFA with three states.
\end{proof}

In parallel with this result, we have computed the first few elements from the sets $\mathcal{L}_i$ of $\{\ta,\ta^2,\tb,\tb^2\}$ for all words of length less than $24$:
$$(\max\{\g(u)\})_{|u|=0,1,2,\ldots}=0,1,2,3,3,4,4,4,4,4,4,5,5,5,6,6,6,6,6,6,7,7,8,8,\ldots$$
Our program that iteratively builds the DFA for the Grundy function did not found any reasonable candidate up to this length. Hence a natural question arises.

\begin{question}
For the game $\{\ta,\ta^2,\tb,\tb^2\}$, is there a set $\mathcal{L}_i$ that is not regular? 
\end{question}

In addition, our program found that new Grundy values regularly appear when the length of the words grows. For example, there are words of length $22$ with a Grundy value of $8$. This correlation between the regularity of $\mathcal{L}_0$ and a finite number of Grundy values has already been established for some rewrite games. Indeed, in the game {\sc duotaire}, as well as for taking-and-breaking games (see Theorem~\ref{thm:l0}), an argument to ensure that $\mathcal{L}_0$ is not regular consists in showing that the Grundy values are not bounded. We wonder whether this property remains true for taking-and-merging games:

\begin{question}\label{q:tam}
Are there taking-and-merging games for which the set $\mathcal{L}_0$ is regular but the Grundy function is not bounded?
\end{question}

Note that in Question~\ref{q:tam}, the converse property does not hold. Indeed, a game for which the Grundy values are bounded does not necessarily has a regular language for $\mathcal{L}_0$. Consider the example of the game $\{\ta^2,\tb^2\}$ detailed in Section~\ref{sec:22}, for which $\mathcal{L}_0$ is proved to be not regular (and where the Grundy values do not exceed $1$).

\section{Winning positions and regular languages}

\newcommand{\tra}{\,{\triangleright_{\!A}}\,}
\newcommand{\trb}{\,{\triangleright_{\!B}}\,}
\newcommand{\tla}{\,{\triangleleft_A}\,}
\newcommand{\tlb}{\,{\triangleleft_B}\,}
\newcommand{\tape}{{\text{\sc Tape}}}
\newcommand{\head}{{\text{\sc Head}}}
\newcommand{\state}{{\text{\sc State}}}

Here we consider slightly more general rewriting rules and show that a very simple problem then become undecidable.
More precisely we consider strongly terminating rewriting games, as defined below. 

\begin{definition}\label{def:strongly terminating}
  A rewriting game~$G$ is called \emph{strongly terminating} if every
  reduction~$\move{u}{v}$ is such that~$|u| > |v|$.
\end{definition}

As the name suggests, a strongly terminating game is such that, from any given starting position, the game is terminating.
For such a game, there is a trivial algorithm computing the Grundy value of a given position,
although in the worst case, this algorithm runs in exponential time with respect to the length of the starting position.
We consider here the following more general problem.

\begin{problem}\label{p.winn-posi-cap-regu}
    Given a strongly terminating game~$G$, and a language~$L$ of ``starting positions'', decide whether there is a $\N$-position for~$G$ belonging to~$L$.
\end{problem}

The main result of this section is the following.

\begin{theorem}\label{wpwr-unde}
  Problem \ref{p.winn-posi-cap-regu} is undecidable, even though the language~$L$ is a star-free regular language.
\end{theorem}

We will prove Theorem~\ref{wpwr-unde} by a reduction from the halting problem of a deterministic Turing machine on the empty word.
It takes indeed the rest of Section~\thesection{}.

\subsection{Instantiation of Problem~\ref{p.winn-posi-cap-regu}}

In the following, we consider a deterministic Turing machine~$T$ defined by%
\vspace*{-.5\topsep}
\begin{itemize}\itemsep=.25\parskip\parskip=.5\parskip
    \item $Q$, the finite set of states;
    \item $q_0\in Q$, the initial state;
    \item $q_{\text{accept}},q_{\text{reject}}\in Q$, the accept and reject state;
    \item $\Gamma$, the finite alphabet of tape symbols;
    \item $\$\in\Gamma$, the left marker of the tape;
    \item $\beta\in\Gamma$, the blank symbol;
    \item $\Sigma\subseteq\Gamma$, the set of input symbols\footnote{Since we  only consider the empty word as input, the set of input symbols is irrelevant.}; and 
    \item $\delta: ((Q\setminus F) \times \Gamma) \rightarrow (Q\times \Gamma \times \{\triangleleft,\triangleright\})$, the partial transition function.
\end{itemize}%
\vspace{-.5\topsep}
We denote by~$F$ the set of halting states, that is:~$F=\{q_{\text{accept}},q_{\text{reject}}\}$.
As usual, we assume that the head of~$T$ is on the symbol \$ at the beginning of a computation.
We also assume for every state~$q$ in~$(Q\setminus F)$ that~$\delta(q,\$)$, if it is defined, is always equal
to~$(r,\$,\triangleright)$ for some state~$r$.
For more details on Turing machine or the Halting Problem, see for instance \cite{Hopcroft,Sipser:2013}.

\medskip

Now, let us define the instance~$(G,L)$ of Problem~$\ref{p.winn-posi-cap-regu}$
to which we reduce the halting of~$T$ on the empty word.

In the following, the first player is called A(lice) and the second one is called B(ob).

First, the alphabet of the game~$G$ is~$Q \uplus \Gamma \uplus\{\#\}\uplus M$, 
where~$Q,\Gamma$ are defined above, where~$\#$ is the `erasable' symbol that will make~$G$ strongly terminating, 
and where
\begin{equation}
           M = \{ \tra,~ \trb,~\tla,~\tlb\}
\end{equation}
is the set of `head symbols', which indicate the position and direction of the head,
as well as the current player (A or B). 

Second, the reductions are defined by equations~\eqref{rule-a-left} to~\eqref{rule-a-win-righ}, below.
For every state~$q$ in~$Q$, the \emph{left-shift reductions} are as follows.
\begin{gather}
    \label{rule-a-left}
    \#\#\#\# \tla q    \quad \movearrow_G \quad    \#\tlb q\, \#\# \\
    \label{rule-b-left}
    \#\tlb q  \quad \movearrow_G \quad  \tla q
\end{gather}
Symmetrically, the \emph{right-shift reductions} are as follows, for every state~$q$ in~$Q$.
\begin{gather}
    \label{rule-a-righ}
    q \tra  \#\#\#\#   \quad \movearrow_G \quad    \#\# q\trb  \#
          \\
    \label{rule-b-righ}
           q \trb \#  \quad \movearrow_G \quad  q \tra
\end{gather}
The \emph{right-transition reductions} are as follows, for every states~$p,q$ in~$(Q\setminus F)$ and for every tape symbol~$\alpha,\gamma$ in~$\Gamma$ such that~$\delta_T(p,\alpha)=(q,\gamma,\triangleleft)$.
\begin{gather}
    \label{rule-a-left-to-left}
    \#\#\#\#\, \alpha \tla p \quad \movearrow_G \quad \# \tlb q \,\#\#\, \gamma \\
    \label{rule-a-righ-to-left}
    \#\#\#\#\, p \tra \alpha \quad \movearrow_G \quad \# \tlb q \,\#\#\, \gamma
\end{gather}
Similarly, if~$\delta_T(p,\alpha)=(q,\gamma,\triangleright)$, 
the \emph{left-transition reductions} are as follows.
\begin{gather}
    \label{rule-a-left-to-righ}
    \alpha \tla p \,\#\#\#\# \quad \movearrow_G \quad \gamma \,\#\#\, q \trb \# \\
    \label{rule-a-righ-to-righ}
    p \tra \alpha \,\#\#\#\#  \quad \movearrow_G \quad \gamma \,\#\#\, q \trb \#
\end{gather}
Finally, the \emph{halting reductions} are as follows, for every states~$p$ in~$(Q\setminus F)$ and for every tape symbol~$c$ in~$\Gamma$ such that~$\delta_T(p,c)=(q,d,x)$,
for some~$q\in F$, $d\in\Gamma$ and~$x\in\{\triangleleft,\triangleright\}$.
\begin{gather}
    \label{rule-a-win-left}
    c \tla p \quad \movearrow_G \quad q
    \\
    \label{rule-a-win-righ}
    p \tra c \quad \movearrow_G \quad q
\end{gather}

Third, the language~$L$ of starting positions is
\begin{equation}
    L = \$\tla q_0 (\# + \beta)^*\quad.
\end{equation}

It may be verified that, thus defined,~$G$ is indeed strongly terminating
and~$L$ is a star-free regular language.

\subsection{Game~$G$ is a zero-player game}

First, easy inductions show the following.
\begin{lemma}\label{lem:at-most-one-head}
    Let us consider the game~$G$ starting from a starting position ${w_0\in L}$.
    Let~$w$ be a later position in a run of the game.
    \vspace{-.5\topsep}
    \begin{itemize}\itemsep=.25\parskip\parskip=.5\parskip
        \item Position~$w$ contains exactly one occurrence of a symbol from~$Q$.
        \item If it is player~A's turn,~$w$ contains no occurrence of~$\trb$ or~$\tlb$,
        and contains exactly one occurrence of either~$\tra$ or~$\tla$.
        \item If it is player~B's turn, then 
        $w$ contains no occurrence of~$\tra$ or~$\tla$
        and~$w$ contains at most one occurrence of a symbol in~$\{\trb,\tlb\}$.
        
        Moreover, if~$w$ contains no head symbol, then the last reduction applied was either~\eqref{rule-a-win-left} or~\eqref{rule-a-win-righ}.
    \end{itemize}
\end{lemma}

It follows immediately that the game is in fact asymmetrical.
The only reduction that player~B can ever apply are~(\ref{rule-b-left}) and~(\ref{rule-b-righ}); while the only reduction that  player A can ever apply are
    the other ones (i.e., rules
    \eqref{rule-a-left}, \eqref{rule-a-righ}, \eqref{rule-a-left-to-left},
    \eqref{rule-a-righ-to-left}, \eqref{rule-a-left-to-righ}, \eqref{rule-a-righ-to-righ}, \eqref{rule-a-win-left} and \eqref{rule-a-win-righ}).
Next lemma states the condition for the game to end.

\begin{lemma}\label{lem:bob-does-not-play}
  Let us consider the game~$G$ starting from a position in~$L$.
  If player~A makes a halting move, then she wins the game.
  Otherwise, player~B always has a move to make afterwards.
\end{lemma}
\begin{proof}
    After applying \eqref{rule-a-win-left} or \eqref{rule-a-win-righ},
    then no symbol in~$\{\tra,\trb,\tla,\tlb\}$ appear in the position;
    hence no reduction can be applied any longer.
    After applying \eqref{rule-a-left}, \eqref{rule-a-left-to-left} or \eqref{rule-a-righ-to-left}, then player~B can always apply~\eqref{rule-b-left}.
    Similarly, after applying \eqref{rule-a-righ}, \eqref{rule-a-left-to-righ} or \eqref{rule-a-righ-to-righ}, then player~B can always apply~\eqref{rule-b-righ}.
\end{proof}

Finally, let us show that~$G$ is a zero-player game, i.e., each move of the game is forced.

\begin{proposition}\label{prop:zero-play-game}
    Let us consider the game~$G$ starting from a position~$w_0$ in~$L$.
    There is a unique sequence of words~$w_1,\ldots,w_n$
    and a unique sequence of reductions~$r_0,\ldots,r_{n-1}$
    such that for every integer~$i$,~$0\leq i < n$, it holds~$\move[r_i]{w_i}{w_{i+1}}$.
    Moreover, player~A wins if and only if~$n>0$ and~$r_{n-1}$ is an instance of~\eqref{rule-a-win-left} or~\eqref{rule-a-win-righ}.
\end{proposition}
\begin{proof}
    From Lemma~\ref{lem:at-most-one-head}, one may see than no position coming 
    from~$w_0$ may ever have more than one head symbol.
    Let~$u$ be a word with only one head symbol~$h$
    and let us show that at most one reduction may be applied to~$u$.
    Indeed, if~$h$ is~$\tlb$ or~$\trb$, only reduction \eqref{rule-b-righ} or \eqref{rule-b-left} may be applied, respectively.
    If~$h=\tla$, the only reductions that may be applied are \eqref{rule-a-left}, \eqref{rule-a-left-to-left}, \eqref{rule-a-left-to-righ}, or \eqref{rule-a-win-left},
    and it is easy to see that if one may be applied the other ones cannot (sometimes because~$T$ was assumed to be deterministic).
%    
%    %
%    Let us assume that~$h=\tla$.
%    %
%    The only reduction rules that may
%    be applied to~$u$ are \eqref{rule-a-left}, \eqref{rule-a-left-to-left}
%    or \eqref{rule-a-left-to-righ}, or \eqref{rule-a-win-left}.
%    %
%    If \eqref{rule-a-left} may be applied, the symbol directly left of~$h$ 
%    is~$\#$, in which case neither of the other reduction rules may be applied.
%    %
%    Otherwise, we assume that the symbol directly left of~$h$ is some~$\alpha$ in~$\Gamma$ and the symbol directly right of~$h$ is some~$p$ in~$Q$ (if it is not the case, no reduction rule may be applied).
%    %
%    Since~$T$ is deterministic, the value of~$\delta_T(p,\alpha)$ entirely determines which rule may be applied.
%    \begin{itemize}
%        \item If~$\delta_T(p,\alpha)$ undefined, no reduction rule may be applied.
%        \item If~$\delta_T(p,\alpha)=(f,\gamma,x)$ for some~$f\in F, \gamma\in\Gamma$ and~$x\in\{\triangleleft, \triangleright\}$, only~\eqref{rule-a-win-left} may be applied (and in only one way).
%        \item If~$\delta_T(p,\alpha)=(q,\gamma,\triangleleft)$ for some~$q\in (Q\setminus F), \gamma\in\Gamma$, only~\eqref{rule-a-left-to-left} may be applied (and in only one way).
%        \item If~$\delta_T(p,\alpha)=(q,\gamma,\triangleright)$ for some~$q\in (Q\setminus F), \gamma\in\Gamma$, only~\eqref{rule-a-left-to-righ} may be applied (and in only one way).
%    \end{itemize}
%    
    A similar reasoning yields for the case~$h=\tra$.
    Applying Lemma~\ref{lem:bob-does-not-play} concludes the proof.
\end{proof}

\subsection{Game~$G$ simulates part of the run of~$T$}

In this section, we define how the current position in~$G$
bears the state and tape of a step of the run of the Turing machine~$T$.

\begin{definition}
    Let~$u=a_0a_1\cdots a_n$ be a word in~$\Gamma^*\big((Q\tra\Gamma)+(\Gamma\tla Q)\big)\Gamma^*$.
    \vspace*{-.5\topsep}
    \begin{itemize}\itemsep=.25\parskip\parskip=.5\parskip
        \item We denote by~$\tape(u)$ the infinite sequence~$v\beta^\omega$,
        where~$v$ is the word in~$\Gamma^*$ resulting from erasing
        from~$u$ each letter that belongs to~$\big(Q \cup \{\tra,\tla\}\big)$.
        \item  We denote by~$\state(u)$ the unique symbol in~$Q$ that appears in~$u$.
        \item We denote by~$\head(u)$ the integer~$j-1$, where~$j$ is such that~$a_j\in\{\tra,\tla\}$ in~$u$.
    \end{itemize}
\end{definition}

Note that the letter at index $\head(u)$ in $\tape(u)$ is exactly the letter pointed at by~$\tra$ or~$\tla$ in~$u$.
For instance if $u=\$\alpha_1\alpha_3\alpha_3p\tra \alpha_2\beta\beta$,
then $\tape(u)=\$\alpha_1\alpha_3\alpha_3\alpha_2\beta^\omega$, $\state(u)=p$
and $\head(u)=4$, that is, the index of~$\alpha_2$ in $\tape(u)$.

\begin{proposition}\label{prop:turi-simu}
    We again take notation of Proposition~\ref{prop:zero-play-game}.
    Let~$\varphi$ be the word morphism erasing the symbols~$\#$.
    Let~$i$ be an even integer, $0\leq i < n$ (that is, a position where it is player A's turn).
    Then, the run of the Turing machine~$T$ on the empty word eventually reaches state~$\state(\varphi(w_i))$ with tape~$\tape(\varphi(w_i))$ and head at position~$\head(\varphi(w_i))$.
\end{proposition}
\begin{proof} 
    By induction on~$i$.
    Case~$i=0$ yields~$\state(w_0)=q_0$,~$\tape(w_0)=\$\beta^\omega$,~$\head(w_0)=0$, that is, the initial setup of the Turing machine~$T$.
    Let~$(i+2)$ be an even integer, $0\leq i < (n-2)$.
    If~$r_i$ is~\eqref{rule-a-left} 
    (resp.~\eqref{rule-a-righ}), then~$r_{i+1}$ is~\eqref{rule-b-left}
    (resp.~\eqref{rule-b-righ}) and~$\varphi(w_{i+2})=\varphi(w_{i})$,
    and induction hypothesis concludes the case.
    Reduction $r_i$ cannot be~\eqref{rule-a-win-left}
    or~\eqref{rule-a-win-righ}, because then player B would have no rule to apply and it would hold~$(i+2)=n$.
    Reduction~$r_i$ is either~\eqref{rule-a-left-to-left}, \eqref{rule-a-righ-to-left}, \eqref{rule-a-left-to-righ} or~\eqref{rule-a-righ-to-righ}.
    We will assume that~$r_i$ is \eqref{rule-a-left-to-righ}; other cases are treated similarly.
    Since we may apply reduction~\eqref{rule-a-left-to-righ}, $w_i$ is equal to $u(\alpha \tla p \,\#\#\#\#)v$ with $u,v\in(\Gamma+\{\#\})^*$ and such that~$\delta_T(p,\alpha)$ is defined and equal to~$(q,\gamma,\triangleright)$ for some~$q\in (Q\setminus F)$ and~$\gamma\in\Gamma$.
    It follows that~$w_{i+2}=u(\gamma\#\#\,q\tra)v$,
    since player A uses reduction~\eqref{rule-a-left-to-righ} and then player B necessarily uses reduction~\eqref{rule-b-righ}.
    We write
    \begin{equation*}
        \state(\varphi(w_{i}))=p
        ~,~
        \head(\varphi(w_{i}))=j
        ~~\text{and}~~
        \tape(\varphi(w_{i}))=u'\alpha v'~,
    \end{equation*}
    where~$u'$ is such that~$|u'|=j$.
    Hence, the following equalities hold.
    \begin{equation*}
        \state(\varphi(w_{i+2}))=q 
       \quad\quad
        \head(\varphi(w_{i+2}))=j+1
        \quad\quad
        \tape(\varphi(w_{i+2}))=u'\gamma v'
    \end{equation*}
    It is exactly the state, tape and head position
    one obtains by applying the transition defined by~$\delta_T(p,\alpha)=(q,\gamma,\triangleright)$
    from the state~$p$, tape~$u'\alpha v'$ and head position~$j$.
\end{proof}

\begin{corollary}\label{cor:forw}
    We again take notation of Proposition~\ref{prop:turi-simu}.
    If~$w_0$ is a winning position for player~A, then~$T$ halts
    on the empty word.
\end{corollary}
\begin{proof}
    Let~$j\in\naturals$,~$\alpha\in\Gamma$, $p\in Q$ and~$u,v\in\Gamma^*$ be such that~$|u|=j$,
    \begin{equation*}
        \state(\varphi(w_{n-1}))=p
        ~,~~
        \head(\varphi(w_{n-1}))=j
        ~~\text{and}~~
        \tape(\varphi(w_{n-1}))=u\alpha v~.
    \end{equation*}
    Since~$w_0$ is a winning position for player~A,~$n$ is odd and~$r_{n-1}$
    is an instance of a halting reduction.
    It follows that~$\delta_T(p,\alpha)$ is defined
    and that its first component is an accepting state.
    Moreover,~$n-1$ is even and we apply Proposition~\ref{prop:turi-simu}:
    the Turing machine~$T$ eventually reaches state~$p$ with its head on~$\alpha$, hence accepts.
\end{proof}

\subsection{From the proper starting position, $G$ may simulate each finite run of~$T$}

Application of left-transition, right-transition or halting reductions (i.e., reductions from \eqref{rule-a-left-to-left} to~\eqref{rule-a-win-righ}),
corresponds to the actual simulation of transitions of~$T$.
Other reductions simply allow to shift the head symbol to the next tape symbol.
Next lemma states a sufficient condition for a `complete' simulation of one (non-halting) transition of~$T$, that is 1) applying the transition and 2) shifting entirely the head to the next tape symbol.

\begin{lemma}
    Let~$\alpha,\gamma,\theta\in\Gamma$ be three tape symbols,
    $p,q\in Q$ be two states,
    $u,v$ be two positions of~$G$
    and~$n$ be a positive integer.
    Then~$u$ reduces to~$v$ in~$2n$ moves in the following cases.
    \vspace*{-.5\topsep}
    \begin{enumerate}\renewcommand{\labelenumi}{{\sf(\roman{enumi})}}
        \itemsep=.25\parskip\parskip=.5\parskip
        \item 
            $\delta_T(p,\alpha)=(q,\gamma,\triangleright)$~,~ $u=\alpha\tla p \,\#^{4n}\theta$~,~  $v=\gamma\,\#^{2n}q\,\tra\theta$
        \item 
            $\delta_T(p,\alpha)=(q,\gamma,\triangleright)$~,~ $u=p\tra\alpha\,\#^{4n}\theta$~,~ $v=\gamma\,\#^{2n}q\,\tra\theta$
        \item 
            $\delta_T(p,\alpha)=(q,\gamma,\triangleleft)$~,~ $u=\theta\,\#^{4n}\,\alpha \tla p$~,~ $v=\theta\tla\,q\#^{2n}\,\gamma$
        \item 
            $\delta_T(p,\alpha)=(q,\gamma,\triangleleft)$~,~ $u=\theta\,\#^{4n}\,p \tra \alpha$~,~ $v=\theta\tla\,q\#^{2n}\,\gamma$
    \end{enumerate}
\end{lemma}
\begin{proof} 
    For item~\textsf{(i)}, apply reduction~\eqref{rule-a-left-to-righ}, then alternatively~$n$ times reduction~\eqref{rule-b-righ} and~$(n-1)$ times reduction~\eqref{rule-a-righ}.
    Proofs of items~\textsf{(ii)},~\textsf{(iii)} and~\textsf{(iv)} are similar.
\end{proof}

In other words, if a transition of~$T$ makes the head shift right (resp.~left) then it will
be executed `completely' in~$G$ if there are~$4n$ consecutive occurrences of symbol~$\#$, for some positive~$n$, right of (resp.~left of) the factor of~$u$ that belongs to~$\big((\Gamma\tra Q)+(Q\tla\Gamma)\big)$.
Then, an induction yields the following.

\begin{proposition}\label{prop:bakw}
    Let us assume that~$T$ halts on the empty word after~$m$ transitions.
    Then, the following word~$w$ is a winning position for~$G$. 
    \begin{equation*}
        w=\$ \tla q_0 \left(\#^{2^{(m+1)}}\beta\right)^m
    \end{equation*}
\end{proposition}

\begin{corollary}\label{cor:bakw}
  If~$T$ halts on the empty word, then~$L$ contains a winning position for~$G$.
\end{corollary}

Finally, Corollaries~\ref{cor:forw} and~\ref{cor:bakw} directly yield Theorem~\ref{wpwr-unde}. Indeed, assume to the contrary that Problem~\ref{p.winn-posi-cap-regu} is decidable, then one would conclude that the halting problem on the empty word is decidable. Nevertheless the latter problem is well-known to be undecidable \cite{Sipser:2013}.
Let us conclude Section~\thesection{} with a conjecture.
\begin{problem}\label{prob:winn-posi-regu}
  Given a strongly terminating game~$\mathcal{G}$, decide whether the winning positions of~$\mathcal{G}$ form a regular language.
\end{problem}

\begin{conjecture}
  Problem \ref{prob:winn-posi-regu} is undecidable.
\end{conjecture}

\section{Perspectives}

Rewrite games open the door to a large field of new interesting questions, as it generalizes a large set of combinatorial games. In the previous sections, we have  given a couple of open problems that we found the most relevant ones in the context of taking-and-merging games. Could they be adapted with an alphabet of a larger size?

Moreover, there are other instances of rewrite games that would make sense to be investigated as their rules can also be expressed with piles of tokens. Consider for example taking-and-merging games where rules of the form $\move{\ta^k}{\tb^\ell}$ are adjoined. Such games can be seen as taking games where tokens have two colors, say black (for $\ta$) and white (for $\tb$). Moves consist in either removing tokens or flipping black tokens (that become white).
In such games, what would the $\mathcal{L}_i$ languages look like? For example, in the game $\{\move{\ta}{\tb},\emove{\ta},\move{\tb}{\varepsilon}\}$, each Grundy language is regular.

\acknowledgements

We would like to thank Idris Ayouaz for his useful computations made on several instances of this game and the reviewer for his useful comments.

\bibliographystyle{abbrv}

\end{document}